\documentclass[reqno]{amsart}
\usepackage[utf8]{inputenc}
\usepackage{amsmath,amssymb,amsthm,enumerate,graphicx}
\usepackage{bbm,bm}
\usepackage{color}
\usepackage{pdfpages}
\usepackage{hyperref}

\newcommand{\ufr}{\mathit{{\it UFR}}}
\newcommand{\llp}{\mathit{{\it LLP}}}
\newcommand{\cp}{\mathit{{\it CP}}}
\newcommand{\Ex}{\mathbb{E}}
\newcommand{\Var}{\mathrm{Var}}

\newcommand{\Cov}{\mathrm{Cov}}

\renewcommand{\phi}{\varphi}
\renewcommand\atop[2]{\genfrac{}{}{0pt}{}{#1}{#2}}

\author{Andreas Lagerås}
\thanks{Andreas Lagerås, AFA Insurance and Stockholm University, Dept.\ Mathematics}
\author{Mathias Lindholm}
\thanks{Mathias Lindholm, Stockholm University, Dept.\ Mathematics}

\title[Smith-Wilson]{Issues with the Smith-Wilson method}

\newtheorem{theorem}{Theorem}

\newtheorem{lemma}{Lemma}

\theoremstyle{remark}
\newtheorem{example}{Example}
\newtheorem*{definition}{Definition}

\begin{document}

\maketitle

\begin{abstract}
The objective of the present paper is to analyse various features of the Smith-Wilson method used for discounting under the EU regulation Solvency II, with special attention to hedging. In particular, we show that all key rate duration hedges of liabilities beyond the Last Liquid Point will be peculiar. Moreover, we show that there is a connection between the occurrence of negative discount factors and singularities in the convergence criterion used to calibrate the model. The main tool used for analysing hedges is a novel stochastic representation of the Smith-Wilson method. Further, we provide necessary conditions needed in order to construct similar, but hedgeable, discount curves.
\vspace{3mm}


\end{abstract}

\section{Introduction}\label{intro}

In the present paper we analyse the mandated method for calculating the basic risk-free interest rate under Solvency II, the so-called Smith-Wilson method. This is an extra- and interpolation method, which is based on a curve fitting procedure applied to bond prices. The technique is described in a research note by Smith and Wilson from 2001, see \cite{SW}. Since \cite{SW} is not publicly available, we have chosen to follow the notation of the European Insurance and Occupational Pensions Authority (EIOPA) given in \cite{E}. The primary aim with the current paper is to present problems with the Smith-Wilson method, especially with regards to hedging interest rate risk. We show analytically that the oscillating behaviour observed numerically by \cite{C} and \cite{V} is always present (Section \ref{hedging}).

Our main theoretical tool is a representation of Smith-Wilson discount factors as expected values of a certain Gaussian process (Section \ref{SW_process_interpretation}). This representation might be useful if one wants to find methods similar to the Smith-Wilson one, without some of its deficiencies (Section \ref{sw_hedgeable}).

With notation from \cite{E}, we have that the discount factor for tenor $t$, when fitted to $N$ prices for zero coupon bonds with tenors $u_1,\dots,u_N$, is
\begin{align}
P(t) &:= e^{-\omega t} + \sum_{j=1}^N\zeta_jW(t,u_j), ~~~t\ge0,\label{hela formeln}
\end{align}
where $\omega := \log(1+\ufr)$, $\ufr$ is the so-called Ultimate Forward Rate, 
\begin{align}
W(t,u_j) :=  e^{-\omega(t+u_j)}\left(\alpha (t\wedge u_j) - e^{-\alpha(t\vee u_j)}\sinh(\alpha(t\wedge u_j))\right),\label{W-funk}
\end{align}
and $\alpha$ is a parameter determining the rate of convergence to the $\ufr$. 

Based on the above it is seen that the $\zeta_j$'s are obtained by solving the linear equation system defined by \eqref{hela formeln} and \eqref{W-funk} given by the specific time points $\{u_j\}_{j=1}^N$. Another name for $u_N$ given in the regulatory framework is the Last Liquid Point ($\llp$), i.e.\ the last tenor of the supporting zero coupon bonds that are provided by the market.

The $\ufr$ is set to 4.2\% for the eurozone. In general, a higher value of $\alpha$ implies faster convergence to $\ufr$. EIOPA \cite[Paragraph 164]{E} has decided that $\alpha$ should be set as small as possible, though with the lower bound 0.05, while ensuring that the forward intensity $f(t) := -\frac{d}{dt}\log P(t)$ differs at most 0.0001 from $\omega$ (defined above) at a certain tenor called the Convergence Point ($\cp$):
\begin{align}\label{conv}
\left|f(\cp) - \omega\right| \le 0.0001.
\end{align}
This optimisation of $\alpha$ can be troublesome to implement numerically since the left hand side of \eqref{conv}, seen as a function of $\alpha$, can have singularities (Section \ref{sing}).

We also point out below, that having the forward yield to converge to a fixed $\ufr$ gives rise to an inconsistency with how the interest rate stress scenarios are specified in Solvency II (Section \ref{parallel}).

The method can also be applied to coupon bearing bonds, or swaps, but there is no loss in generality in considering only zero coupon bonds. The generalisation is particularly simple since a coupon bearing bond can be seen as a linear combination of zero coupon bonds and the Smith-Wilson method is linear in bond prices.

We also note that the market data that is used as input for the Smith-Wilson method should undergo a credit adjustment. This is nothing that we will specify further, but refer the reader to \cite{E} and merely state that this adjustment is of no relevance for the results below. If anything, the variable credit adjustment will make hedging even harder.

\section{Notation}

For later convenience we will here state relevant abbreviations and notation:
\begin{itemize}
\item[$\llp$] is the Last Liquid Point for where the zero coupon bond market support ends.
\item[$\ufr$] is the Ultimate Forward Rate, i.e.\ 4.2\% for most currencies.
\item[$\omega$] is the continuously compounded ultimate forward rate, i.e.\ $\omega = \log(1+\ufr)$.
\item[$\cp$] is the Convergence Point where the $\ufr$ should be reached.
\item[$\alpha$] is the mean reversion parameter that determines the rate of convergence to the $\ufr$.
\item[$\bm{u}$] is a vector with tenors of the market zero coupon bonds.
\item[$\bm{p}$] is a vector of observed zero coupon prices at times to maturities $\bm{u$}, that is $\bm{p}=(p_1,\dots,p_{\llp})'$.
\item[$\bm{r}$] is a vector of observed zero coupon spot rates at times to maturities $\bm{u$}, that is $\bm{r}=(r_1,\dots,r_{\llp})'$, i.e.\ $\bm{p}_i := \bm{p(r)}_i = e^{-r_iu_i}$.
\item[$H(s,t)$] is the following function:
\[H(s,t) := \alpha (s\wedge t) - e^{-\alpha (s\vee t)}\sinh(\alpha (s\wedge t)).\]
\item[$W(s,t)$] is defined as $W(s,t) := e^{-\omega(s+t)}H(s,t)$. 
\item[$\bm{Q}$] is a diagonal matrix with $\bm{Q}_{ii} = e^{-\omega u_i} =: \bm{q}_i, i = 1,\ldots,\llp$.
\item[$\bm{H}$] is a matrix with elements $\bm{H}_{ij} = (H(u_i,u_j))_{ij}$.
\item[$\bm{W}$] is a matrix with elements $\bm{W}_{ij} = (W(u_i,u_j))_{ij}$. Note that $\bm{W} = \bm{QHQ}$.
\item[$W(t,\bm{u})$] is defined as $W(t,\bm{u}) := (W(t,u_1),\ldots,W(t,u_\llp))'$. $H(t,\bm{u})$ is defined analogously.
\item[$\bm{b}$] is the solution to the equation $\bm{p} = \bm{q} + \bm{QHQb}$ (note that this is the zero coupon case).
\item[$\sinh {[}\alpha \bm{u}'{]}$] denotes $\sinh(\,\cdot\,)$ applied component-wise to the vector $\alpha \bm{u}'$.
\item[$P(t)$] is the discount function at $t$ that surpress the dependence on the market support, i.e.\ $P(t) := P(t;\bm{p}(\bm{r})) \equiv P(t;\bm{p}) \equiv P(t;\bm{r})$.
\item[$P^c$] is the present value of the cash flow $\bm{c}$ w.r.t.\ Smith-Wilson discounting using $P(t)$, i.e.\ $P^c := \sum_t c_tP(t)$. Hence, $P^c := P^c(t;\bm{p}(\bm{r})) \equiv P^c(t;\bm{p}) \equiv P^c(t;\bm{r})$.
\end{itemize}

\section{Representing Smith-Wilson discount factors}\label{SW_process_interpretation}

The problems with the Smith-Wilson method that will be highlighted in later sections are centered around problems regarding hedging. In order to understand this in more detail we have found that the representation of the method from \cite{AL} and \cite{L} will prove useful. We will now give a full account of how the extrapolated discount factors of the Smith-Wilson method can be treated as an expected value of a certain stochastic process:

Let $\{X_t: t\geq 0\}$ be an Ornstein-Uhlenbeck process with $dX_t = -\alpha X_t dt + \alpha^{3/2}dB_t$, where $\alpha>0$ is a mean reversion parameter, and $X_0\sim N(0,\alpha^2)$ independent of $B$, and let $\bar X_t := \int_0^t X_s ds$ and $Y_t := e^{-\omega t}(1+\bar X_t)$. Given this we can state the following theorem:

\begin{theorem}\label{thm_repr}
$P(t) = \Ex[Y_t|Y_{u_i} = p_{u_i}; i=1,\dots,N]$.
\end{theorem}
In other words: the Smith-Wilson bond price function can be interpreted as the conditional expected value of a certain {\it non-stationary} Gaussian process. Note that $\alpha$ will govern {\it both} mean reversion and volatility.

Since $\{Y_t: t \geq 0\}$ is a Gaussian process we have that $P(t)$, being a conditional expected value, is an affine function of $\bm{p}$:
\begin{equation}\label{P_linj}
P(t) = \Ex[Y_t] + \Cov[Y_t,\bm{Y}]\Cov[\bm{Y},\bm{Y}]^{-1}(\bm{p}-\Ex[\bm{Y}]) =:  \beta_0 + \boldsymbol{\beta}'\bm{p},
\end{equation}
where $\beta_0$ and $\boldsymbol{\beta}$ are functions of $t$, but not $\bm{p}$, if $\alpha$ is considered a fixed parameter. If $\alpha$ is set by the convergence criterion, $\beta_0$ and $\boldsymbol{\beta}$ are functions of $\bm{p}$.

The main aim of this paper is to analytically show problems inherent in the Smith-Wilson method which will affect hedging of liabilities. From this perspective it is evident that the re-formulation of the bond price function according to Equation \eqref{P_linj} will prove useful, and in particular the behaviour of the $\beta$'s will be of interest:
\begin{theorem}\label{thm_signs}
If $t>u_N$, $\mathrm{sign}(\beta_i)=(-1)^{N-i}$ for $i=1,\dots, N$.
\end{theorem}
This has peculiar consequences for hedging interest rate risk. The proofs of Theorem \ref{thm_repr} and \ref{thm_signs} are given in Section \ref{proofs}.

\section{Problems with the Smith-Wilson method}

There are a number of problems with the Smith-Wilson method. Some of these were known early and can be found in \cite{E2}. Here we list some of the problems, and we start with a serious one regarding hedging.

\subsection{Hedging}\label{hedging}

If you have a liability, i.e.\ a debt, of 1 unit of currency with tenor $t$, its market value is that of a default free zero coupon bond with the same tenor, since if you buy the zero coupon bond you know that you will be able to pay your debt no matter what happens with interest rates. This is the essence of market valuation of liabilities and also of hedging. If you have liabilities with several tenors, you could theoretically match them by buying the zero coupon bonds with the same tenors.

However, in practise, there are not enough bonds available for longer tenors. This is one reason for the need of an extrapolation method such as Smith-Wilson so that liabilities with large tenors are priced with a model rather than the non-existent market.

The Smith-Wilson method interpolates market prices for a given set $u_1,\dots,u_N$, i.e.\ the model price equals the market price for these tenors. Liabilities with tenors in this set could therefore be perfectly hedged by buying the corresponding amounts of zero coupon bonds.

Liabilities with tenors outside the set of market tenors, have present values that are truly model based. Consider a liability of 1 unit of currency with tenor $t$ which has present value $P(t)$ according to the Smith-Wilson method.

By Equation \eqref{P_linj}, $P(t)$ is affine in the prices of the zero coupon bonds for tenors $u_1,\dots,u_N$. Thus, Equation \eqref{P_linj} gives the recipe for a ``perfect'' hedge: own $\beta_i$ units of the zero coupon bond with time to maturity $u_i$ for $i=1,\dots,N$, and have the amount $\beta_0$ in cash. However the values of the zero coupon bonds fluctuate, the combined portfolio will have the value $P(t)$.\footnote{This only holds assuming $\alpha$ is constant. If the changing market prices force a change of $\alpha$ the hedge is no longer perfect. For moderate changes of the yield curve, the hedge should still perform well.} 

The Smith-Wilson method is most easily expressed in terms of prices rather than yields, but we note that the idea of matching $\beta$'s is essentially the same thing as matching \emph{key rate durations} where the key rates are all market rates used to construct the Smith-Wilson curve. Key rate durations, or rather dollar values of a basis point (which go by abbreviations such as BPV or DV01), are the preferred measure of interest sensitivities and hedge construction in \cite{V} and \cite{C}.

Alas, by Theorem \ref{thm_signs}, the recipe for the perfect hedge is quite strange when you actually try to procure the ingredients: If $t>u_N$ you will need a positive amount of the zero coupon bond at $u_N$, a \emph{negative} amount of the zero coupon bond at $u_{N-1}$, then again a positive amount of the one at $u_{N-2}$, etc. 

This oscillating behaviour was observed empirically for some tested yield curves by \cite{V} and \cite{C}, but as Theorem \ref{thm_signs} show that the pattern is present for \emph{all} Smith-Wilson curves without exception. 
 	
It is also worth noting that a hedge constructed according to the above procedure will due to the sign-changes have a sum of the absolute values of the exposures that is {\it larger} than the present value of the liabilities that one wants to hedge.

Another issue with the hedge is that as soon as time passes, say with a month $\Delta t$, $P(t)$ is calculated with new zero coupon bonds at $u_1,\dots,u_N$, whereas our portfolio has bonds with maturities $u_1-\Delta t, \dots, u_N -\Delta t$. If the signs of holding amounts had all been positive, this might not had been such a big practical issue. In that case one could conceivably have changed the portfolio weights little by little and still have had an acceptable hedge.

Theorem \ref{thm_signs} implies that whatever holding one has with a particular maturity date must be sold and changed into the opposite exposure in the time span of one year in the case of when all time to maturities $u_1,\dots,u_N$ are one year apart. For many currencies $u_N - u_{N-1} = 5$ years. Even in this case the turnover would be impractically large.

\begin{example}\label{mat30}
Consider an initial EUR curve with observed market rates for tenors 1, 2, \dots,  10, 12, 15, and 20 all equal to 4.2\%. This means that the Smith-Wilson curve is flat at 4.2\% for all tenors. A liability of 100 EUR with maturity 30 years, i.e.\ in the extrapolated part of the curve will have the present value $100\cdot 1.042^{-30} \doteq 29$. The discount factor $P(30) =  \beta_0 + \boldsymbol{\beta}'\bm{p}$ where $\beta_0=0$, $\beta_i \doteq 0.00$ for $i=1,\dots,6$, and
\vspace{2mm}

{\footnotesize
\begin{tabular}{r|rrrrrrr}
$i$ & 7 & 8 & 9 & 10 & 12 & 15 & 20 \\
\hline
$\beta_i$ & 0.01 & -0.05 & 0.19 & -0.38 & 0.76 & -1.64 & 1.96 \\
$\beta_ip_i$ & 1 & -3 & 13 & -26 & 47 & -88 & 86 
\end{tabular}
}
\vspace{1mm}

Note that the positions in the zero coupon bonds at the last three maturities: 46, -88 and 86, are considerably larger in absolute value than the present value 29 of the liability.
\end{example}

In general the liabilities of an insurance company do not all come due the same date, and if we have undiscounted liabilities $c_t$ with time to maturity $t=1,2\dots$, the present value of all liabilities is $P^c=\sum_t c_tP(t)$. One can still find the hedge by Equation \eqref{P_linj} since $\sum_t c_tP(t) =  \sum_tc_t\beta_0(t) + \sum_tc_t\boldsymbol{\beta}(t)'\bm{p} =: \beta_0^c + {\boldsymbol{\beta}^c}'\bm{p}$.

\begin{example}
Consider the same curve as is Example \ref{mat30}, but let the liability cash flow equal $10/ 1.10^k$ for $k=1,2,\dots$. The undiscounted value of the liabilities is 100 EUR and the present value is 68. The weighted average time to maturity is 11 years and the Macaulay duration is about 7 years. The present value of liabilities discounted by extrapolated yields, i.e.\ longer than 20 years, is 4.47 which is less than 7\% of the total present value. These liabilities will contribute to the oscillating behaviour of the hedge. A priori one might think that the amount is such a small part of the total that the hedge of the overall liabilities would consist of only positive exposures, but then one would be mistaken, since the penultimate position is negative:
\vspace{1mm}

{\footnotesize
\begin{tabular}{r|rrrrrrrrrrrrr}
$i$ & 1 & 2 & 3 & 4 & 5 & 6 & 7 & 8 & 9 & 10 & 12 & 15 & 20 \\
\hline
$\beta_i^c$    & 9 & 8 & 8 & 7 & 6 & 6 & 5 & 4 & 5 & 2 & 15 & -8 & 29 \\
$\beta_i^cp_i$ & 9 & 8 & 7 & 6 & 5 & 4 & 4 & 3 & 4 & 1 &  9 & -4 & 13
\end{tabular}
}

We note that this pattern is also seen in \cite[Fig.\ 2]{C}.
\end{example}

\subsection{Negative discount factors}\label{sec_neg_df}

Discount factors extrapolated by the Smith-Wilson method may become negative when the market curve has a steep slope for high tenors, i.e.\ the last market forward rates are high. This has been noted by supervisory authorities, e.g.\ \cite{E2,FT}.

\begin{example}\label{ex_neg_df} A simple, not completely unrealistic example, is market rates $r_t = t$\% for tenors $t=1$, 2, \dots,  10, 12, 15 and 20. With $\alpha = 0.22$ we have convergence in the sense of Equation \eqref{conv} at $\mathit{CP} = 60$, and negative discount factors for all tenors larger than 24 years.
\end{example}

This is nonsense and clearly very undesirable. A single set of market inputs and Smith-Wilson curve output may be checked manually, but when market rates are simulated or drawn from an Economic Scenario Generator, and the Smith-Wilson method is applied to them, checks must be automated. EIOPA does not specify how to amend the method when discount factors become negative. 

One solution is to increase $\alpha$ even further. In Example \ref{ex_neg_df} above, it suffices to increase $\alpha$ to 0.32 to avoid negative discount factors. This is nothing strange, since an increase in $\alpha$ corresponds to an increase in the speed of mean reversion, and hence an increase in the the stiffness of the curve. In order to see that this is always possible, one can argue as follows: for large values of $\alpha$ and $t \ge s$ the function $W(s,t) \sim \alpha (s \wedge t)$, which corresponds to the covariance function of Brownian motion, implying that the conditional expected value, i.e.\ the discount function at $t$, will be close to $e^{-\omega t}$. Since this holds for any $t\ge s$, it will in particular hold for $t = \cp \ge u_i$. Thus, increasing $\alpha$ will eventually make the discount factors become positive.

Note that there is no contradiction between negative discount factors and convergence of continuously compounded forward rates, i.e.\ the forward intensities. This because the latter are gradients of the corresponding bond prices, or alternatively put, discount factors. The problem with negative discount factors is rather that they can not be represented as any {\it real}, as a converse to imaginary, spot rates.

\subsection{Tricky to find \texorpdfstring{$\bm{\alpha}$}{alpha}: negative discount factors revisited}\label{sing}

Another problem with the calibration of $\alpha$ is that there may be singularities in the domain where $\alpha$ is optimised:

\begin{example}\label{ex_alpha_singularity}
Consider a parametrisation according to the Swedish market, i.e.\ $\llp = 10, \cp = 20$ and $\omega = \log(1+4.2\%)$ together with the following zero coupon spot rates: 2\%, 2.2\%, 2.4\%, 3\%, 3.2\%, 4\%, 5\%, 6\%, 6.25\%, 7.5\%, defined for maturities 1 to 10. In Figure \ref{fig_alpha_singularity} it is clear that there is a singularity in terms of the calibration criterion defined by EIOPA in \cite[Paragraphs 160--166]{E}.
\end{example}

To understand this better, see to the convergence criterion defined by EIOPA, i.e.\ \cite[Paragraphs 160--166]{E}, that can be expressed as
\begin{align}
g(\alpha) := |h(\alpha)| = |f(\cp) - \omega| = \frac{\alpha}{|1-\kappa e^{\alpha\cp}|},\label{conv g}
\end{align}
where 
\begin{align}
\kappa := \frac{1+\alpha \bm{u}' \bm{Qb}}{\sinh[\alpha \bm{u}'] \bm{Qb}}.
\end{align}
In order for a singularity to arise it is hence necessary that
\[
1-\kappa e^{\alpha\cp} \equiv 0,
\]
which is equivalent to
\[
1 + (\alpha \bm{u}' -  e^{-\alpha\cp} \sinh[\alpha \bm{u}'] )\bm{Qb} = 0.
\]
Now note that
\[
P(\cp) = e^{-\omega \cp} \left(1 + (\alpha \bm{u}' -  e^{-\alpha\cp} \sinh[\alpha \bm{u}'] )\bm{Qb}\right),
\]
that is $g(\alpha)$ from \eqref{conv g} is only singular iff $P(\cp) \equiv 0$. Moreover, it holds that $h(\alpha)$ from $\eqref{conv g}$ satisfy $h(\alpha) \le 0$ iff $P(\cp) \le 0$ and $h(\alpha) > 0$ iff $P(\cp) > 0$. Consequently a singularity in the domain where $\alpha$ is optimised can only occur if the input market spot curve will result in a singularity for some of the eligible $\alpha$ values, i.e.\ for some $\alpha \ge 0.05$ according to EIOPA's specification. 

To conclude, from the previous Section \ref{sec_neg_df} we know that there may be situations when we have convergence according to $g(\alpha)$ from \eqref{conv g}, but where the resulting discount factors are negative. We have now learned that if we want to solve this situation by increasing $\alpha$, which we know will work, we need to pass a singularity. Furthermore, we now know that even if the optimisation algorithm will not end up in this pathological situation, as soon as the input market spot rate may give rise to negative discount factors for {\it some} $\alpha \ge 0.05$, there will be a singularity that must be avoided by the optimisation algorithm that one uses.

\begin{figure}
\includegraphics[width=0.5\textwidth]{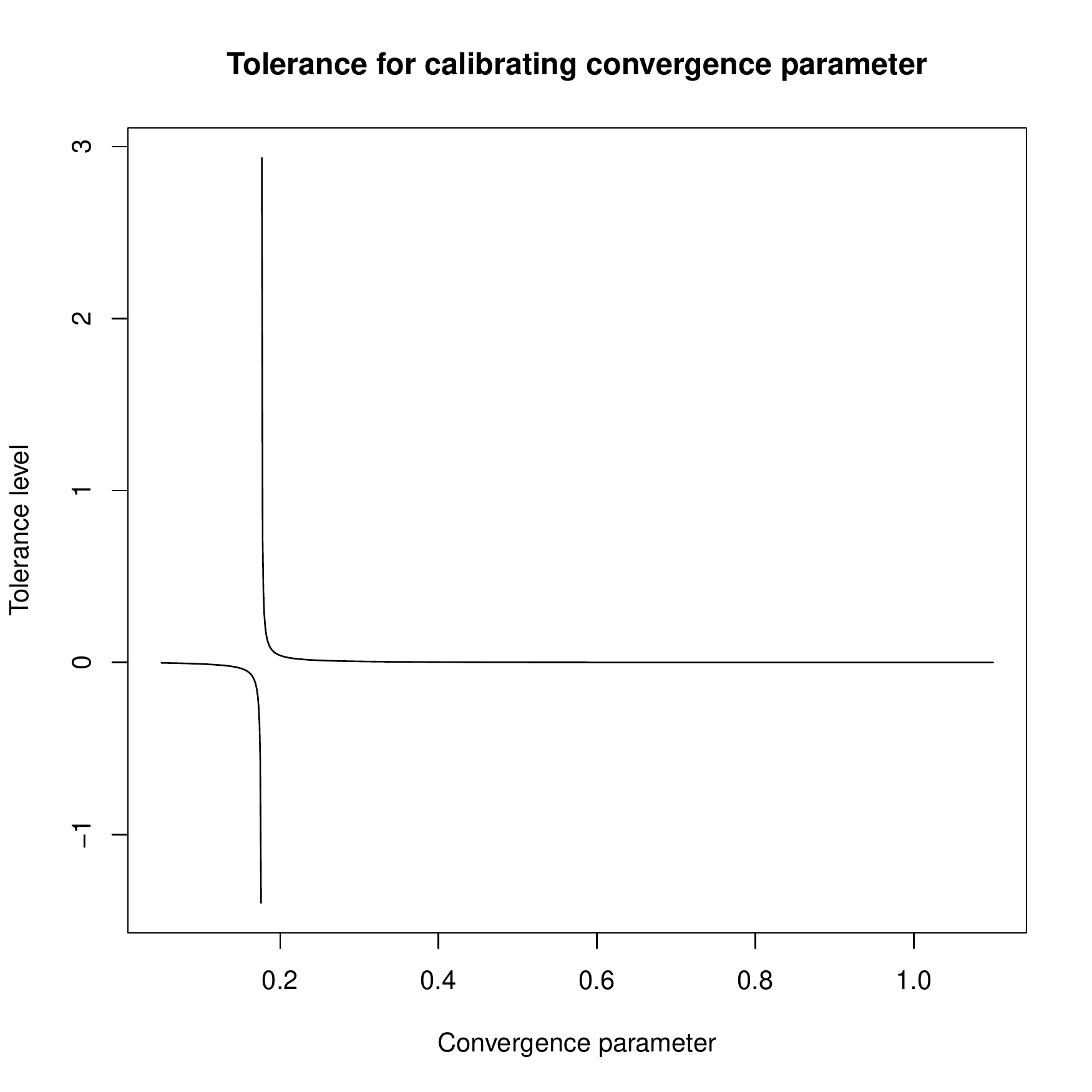}
\caption{Plot of tolerance function $h(\alpha)$ from \eqref{conv g} used for calibrating $\alpha$ for parameters according to Example \ref{ex_alpha_singularity}.}\label{fig_alpha_singularity}
\end{figure}

\subsection{What do the EIOPA stresses really mean?}\label{parallel}

In the Solvency II regulation the interest rate stress is defined as follows: let $r_t$ denote the basic risk-free spot rate for a zero coupon bond with $t$ time units to maturity, and let $r_t^s$ denote its stressed counterpart. According to the regulation the stressed spot rate is given by
\[
r_t^s := r_t(1+s_t),
\]
where $s_t$ is a pre-specified constant that is positive or negative depending on whether one considers increasing or decreasing interest rates. The shift is at least $\pm 0.20$ (for tenors longer than 90 years). For long tenors, such as the $\cp$, the shift is essentially parallel, i.e.\ it does not vary much with $t$. This means that the forward rate essentially shifts with the same factor as the rate, and in particular the forward rate at the $\cp$ shifts with a substantial amount: approximately more than $\pm0.2\cdot 4.2\%$. This goes against the whole idea of the ultimate forward rate being a \emph{constant}.

Note that there is no easy way to fix this: One alternative could be to only stress the spot rates up until the $\llp$ and thereafter re-calibrate $\alpha$ with an unchanged $\ufr$. Another alternative could be to instead stress all forward rates up until the $\ufr$, and thereafter re-calculate the implied spot rates, hence without the need of re-calibrating $\alpha$.

\section{Proofs}\label{proofs}

\begin{proof}[Proof of Theorem \ref{thm_repr}]
First note that the $Y_t$ has mean function $\Ex[Y_t] = e^{-\omega t}(1 + \Ex[\bar X_t]) = e^{-\omega t}$. We intend to show that $Y_t$ has the Wilson function $W(s,t)$ as its covariance function. Let $0\leq s\leq t$. The non-stationary Ornstein-Uhlenbeck process $X_t = X_0e^{-\alpha t} + \alpha^{3/2}\int_0^t e^{-\alpha(t-s)}dB_s$ has
\begin{align*}
\Cov[X_s,X_t] &=: K(s,t) \\
&= e^{-\alpha(s+t)}\Cov\left[X_0 + \alpha^{3/2}\int_0^s e^{\alpha u}dB_u, X_0 + \alpha^{3/2}\int_0^t e^{\alpha u}dB_u\right] \\
&= e^{-\alpha(s+t)}\left(\Var[X_0] + \alpha^3\int_0^s e^{2\alpha u}du\right) \\
&= e^{-\alpha(s+t)}\left(\alpha^2 + \frac{\alpha^2}{2}(e^{2\alpha s}-1)\right) \\
&= \alpha^2 e^{-\alpha t}\cosh(\alpha s).
\end{align*}
The integrated process $\bar X_t$ thus has,
\begin{align*}
\Cov[\bar X_s,\bar X_t] &=: H(s,t) = \iint\limits_{\atop{0\leq u\leq s}{0\leq v\leq t}}K(u,v)\,dudv \\
&= \bigg(\iint\limits_{0\leq u\leq v \leq s}+\iint\limits_{0\leq v \leq u\leq s}+\iint\limits_{0\leq u\leq s \leq v\leq t}\bigg)K(u,v)\,dudv \\
&= \bigg(2\iint\limits_{0\leq u\leq v \leq s}+\iint\limits_{0\leq u\leq s \leq v\leq t}\bigg)K(u,v)\,dudv \\
&= 2\int_0^s\alpha e^{-\alpha v}\bigg(\int_0^v\alpha\cosh(\alpha u)\,du\bigg)\,dv + \int_s^t\alpha e^{-\alpha v}dv\int_0^s\alpha\cosh(\alpha u)\,du \\
&= 2\int_0^s\alpha e^{-\alpha v}\sinh(\alpha v)\,dv + (e^{-\alpha s}-e^{-\alpha t})\sinh(\alpha s)\\
&= \alpha s - e^{-\alpha t}\sinh(\alpha s),
\end{align*}
and we arrive at the covariance function
$$
\Cov[Y_s,Y_t] = e^{-\omega(s+t)}\Cov[\bar X_s,\bar X_t] = e^{-\omega(s+t)}(\alpha s - e^{-\alpha t}\sinh(\alpha s)) = W(s,t).
$$
Write $\bm{Y} := (Y_{u_1}, \dots, Y_{u_N})'$ and note that since $Y_t$ is a Gaussian process
\begin{align*}
\Ex[Y_t|\bm{Y}=\bm{p}] &= \Ex[Y_t] + \Cov[Y_t,\bm{Y}]\Cov[\bm{Y},\bm{Y}]^{-1}(\bm{p}-\Ex[\bm{Y}]) \\
&= e^{-\omega t} + \Cov[Y_t,\bm{Y}]\boldsymbol{\zeta} \\
&= e^{-\omega t} + \sum_{i=1}^N\Cov[Y_t,Y_{u_i}]\zeta_i \\
&= e^{-\omega t} + \sum_{i=1}^N W(t,u_i)\zeta_i,
\end{align*}
as desired, where we have identified $\bm{\zeta}:=(\zeta_1,\dots,\zeta_N)':=\Cov(\bm{Y},\bm{Y})^{-1}(\bm{p}-\Ex[\bm{Y}])$.
\end{proof}

For the proof of Theorem \ref{thm_signs} we will use the matrix property \emph{total positivity}, see, e.g., \cite{K}. For an $n$-dimensional matrix $\mathbf{A}=(a_{ij})_{i,j\in\{1,\dots,n\}}$, let $\mathbf{A}[\mathcal{I},\mathcal{J}]:=(a_{ij})_{i\in\mathcal{I},j\in\mathcal{J}}$ be the submatrix formed by rows $\mathcal{I}$ and columns $\mathcal{J}$ from $\mathbf{A}$, and let $\mathbf{A}[-i,-j]=(a_{kl})_{k\neq i,l\neq j}$ be the submatrix formed by deleting row $i$ and column $j$.

Recall that a \emph{minor} of order $k$ of a matrix $\mathbf{A}$ is the determinant $\det\mathbf{A}[\mathcal{I},\mathcal{J}]$, where the number of elements in both $\mathcal{I}$ and $\mathcal{J}$ is $k$.

\begin{definition}[Total positivity] A $n\times n$ matrix is said to be \emph{totally positive} (TP) if all its minors of order $k = 1,\ldots,n$ are non-negative.
\end{definition}

For functions of two arguments, say $f(s,t)$, which we call \emph{kernels}, there are parallel definitions of total positivity, viz.\ $f(s,t)$ is said to be totally positive if the matrices $(f(s_i,t_j))_{i,j\in\{1,\dots,n\}}$ are totally positive for all $n$ and $s_1<\cdots<s_n$ and $t_1<\cdots<t_n$.

We use the following notation for determinants of matrices constructed by kernels:
$$
f\begin{pmatrix}s_1 & \dots & s_n \\ t_1 & \dots & t_n \end{pmatrix} := \det\, (f(s_i,t_j))_{i,j\in\{1,\dots,n\}}
$$

\begin{lemma}\label{lemma_tp}
The covariance function $K(s,t)= \alpha^2 e^{-\max(s,t)}\cosh(\alpha \min(s,t))$ of the Ornstein-Uhlenbeck process $\{X_t:t\geq 0\}$ is totally positive.
\end{lemma}

\begin{proof}[Proof of Lemma \ref{lemma_tp}]
Let $f(t):= \alpha \cosh(\alpha t)$ and $g(t):= \alpha e^{-\alpha t}$. We note that $h(t):= f(t)/g(t) = (1+e^{2\alpha t})/2$ is increasing and can thus apply \cite[Example I.(f), p.\ 213]{A} and conclude that $K$ is totally positive.
\end{proof}

We will need the following observation about continuous non-negative functions.

\begin{lemma}\label{lemma_f_pos}
Let $f:\mathbb{R}^m\to\mathbb{R}$ be a non-negative continuous function and let $A$ be an $m$-dimensional box, which might include some or none of its boundary $\partial A$. If $f$ is positive somewhere on $\partial A$, then $\int_A f(\bm{x})d\bm{x}$ is positive.
\end{lemma}

\begin{proof}[Proof of Lemma \ref{lemma_f_pos}]
By assumption, there exists an $\epsilon>0$ and an $\bm{x}_0\in\partial A$ such that $f(\bm{x}_0)>2\epsilon>0$. Since $f$ is continuous, there exists a $\delta>0$ such that $f(\bm{x})>\epsilon$ for all $\bm{x}\in B_{\delta}(\bm{x}_0):=\{\bm{x}\in\mathbb{R}^m: |\bm{x}-\bm{x}_0|<\delta\}$. Since the intersection between the ball $B_{\delta}(\bm{x}_0)$ and $A$ has positive $m$-dimensional volume and $f$ is non-negative in general, we get
$$
\int_A f(\bm{x})d\bm{x} = \int_{A \cap B_{\delta}(\bm{x}_0)} f(\bm{x})d\bm{x} + \int_{A \setminus B_{\delta}(\bm{x}_0)} f(\bm{x})d\bm{x} > \epsilon \int_{A \cap B_{\delta}(\bm{x}_0)} 1 d\bm{x} + 0 > 0
$$
\end{proof}

We are now ready to prove Theorem \ref{thm_signs}.

\begin{proof}[Proof of Theorem \ref{thm_signs}]
Let us fix $t\equiv u_{N+1}>u_N$. We note that the sign of the regression coefficient $\beta_i$ equals that of the partial correlation coefficient of $Y_{u_{N+1}}$ and $Y_{u_i}$ given all other $Y_{u_j}$, $j\in\{1,\dots,N\}\setminus\{i\}$. This partial correlation coefficient equals that of $\bar X_{u_{N+1}}$ and $\bar X_{u_i}$ given all other $\bar X_{u_j}$, since the two processes have the same correlation matrix. Let us call this partial correlation coefficient $p_{i,N+1}$, and let $\bar{\bm{X}} := (\bar{X}_{u_1},\dots, \bar{X}_{u_{N+1}})'$.

The partial correlation coefficient has the opposite sign to the element on row $i$ and column $N+1$ in the inverse of the covariance matrix $\mathbf{H}:=\Cov[\bar{\bm{X}},\bar{\bm{X}}]$: With $\mathbf{B} = (b_{ij})_{i,j=1,\dots,N+1} := \mathbf{H}^{-1}$, $\mathrm{sign}(p_{i,N+1})=-\mathrm{sign}(b_{i,N+1})$.

By Cramer's rule
$$
b_{i,N+1} = (-1)^{i+N+1}\frac{\det \mathbf{H}[-(N+1),-i]}{\det \mathbf{H}} = (-1)^{N+1-i}\frac{\det \mathbf{H}[-i,-(N+1)]}{\det \mathbf{H}},
$$
since $\mathbf{H}$ is a covariance matrix and thus symmetric. The determinant of $\mathbf{H}$ is positive and therefore
$$
\mathrm{sign}(p_{i,N+1}) = -\mathrm{sign}(b_{i,N+1}) = (-1)^{N-i}\mathrm{sign}( \det \mathbf{H}[-i,-(N+1)])
$$
We thus need to prove that $\det \mathbf{H}[-i,-(N+1)] > 0$ for $i=1,\dots,N$. With $H(s,t)$ being the covariance function of $\bar X_t$, we have
$$
\det \mathbf{H}[-i,-(N+1)] = H\begin{pmatrix} u_1 & \dots & u_{i-1} & u_{i+1} & \dots & u_{N+1} \\ u_1 & \dots & u_{i-1}     & u_i     & \dots & u_N    \end{pmatrix}.
$$
We can write the kernel $H$ as a double integral of the kernel $K$:
$$
H(s,t) = \iint\limits_{\atop{0\leq v\leq s}{0\leq w\leq t}} K(v,w)dvdw = \iint\limits_{\atop{0\leq v}{0\leq w}}\underbrace{\mathbbm{1}\{s\geq v\}}_{=:L(s,v)}K(v,w)\underbrace{\mathbbm{1}\{w\leq t\}}_{=:R(w,t)}dvdw.
$$
The kernel $L$ produces matrices that have ones below a diagonal, and zeros above. The kernel $R$ is similar with ones above a diagonal. Their determinants therefore equal either zero or one, and they equal one only if the ones are on the main diagonal of the matrix.
\begin{align*}
L\begin{pmatrix} s_1 & \dots & s_n \\t_1 & \dots & t_n \end{pmatrix}&=\mathbbm{1}\{t_1\leq s_1 < t_2 \leq s_2 < \cdots < t_n \leq s_n\} \\
R\begin{pmatrix} s_1 & \dots & s_n \\t_1 & \dots & t_n \end{pmatrix}&=\mathbbm{1}\{s_1\leq t_1 < s_2 \leq t_2 < \cdots < s_n \leq t_n\}
\end{align*}

We use this in the continuous version of the Cauchy-Binet formula \cite[Eq.\ (3.1.2)]{K} and obtain
\begin{align}
H&\begin{pmatrix} u_1 & \dots & u_{i-1} & u_{i+1} & \dots & u_{N+1} \\
                  u_1 & \dots & u_{i-1} & u_i     & \dots & u_N    \end{pmatrix} = \notag\\
&= \idotsint\limits_{\atop{0\leq v_1<\cdots<v_N}{0\leq w_1<\cdots<w_N}}
L\begin{pmatrix} u_1 & \dots & u_{i-1} & u_{i+1} & \dots & u_{N+1} \\
                 v_1 & \dots & v_{i-1} & v_i     & \dots & v_N    \end{pmatrix}
H\begin{pmatrix} v_1 & \dots & v_N \\
                 w_1 & \dots & w_N \end{pmatrix}
R\begin{pmatrix} w_1 & \dots & w_N \\
                 u_1 & \dots & u_N \end{pmatrix}d\bm{v}d\bm{w} \notag\\
&= \idotsint\limits_{A}H\begin{pmatrix} v_1 & \dots & v_N \\w_1 & \dots & w_N \end{pmatrix}d\bm{v}d\bm{w}, \label{cb_int}
\end{align}
where 
\begin{align*}
A := \{&0\leq v_1 \leq u_1 < \dots < v_{i-1}\leq u_{i-1} <v_i\leq u_{i+1} < \cdots < v_N \leq u_{N+1},\\
&0\leq w_1 \leq u_1 < \cdots < w_N \leq u_N\}
\end{align*}
Since $K(s,t)$ is a continuous function, and the determinant $\det (a_{ij})_{i,j=1,\dots,n}$ is a continuous function of all elements $a_{ij}$, $i,j = 1,\dots,n$, the determinant $K(\begin{smallmatrix}s_1 & \dots & s_n\\t_1 & \dots & t_n \end{smallmatrix})$ is continuous as a function of $s_1\dots,s_n,t_1,\dots,t_n$. By Lemma \ref{lemma_tp} we also know that $K(\begin{smallmatrix}s_1 & \dots & s_n\\t_1 & \dots & t_n \end{smallmatrix})$ is a non-negative function. Furthermore we have $K(\begin{smallmatrix}u_1 & \dots & u_N\\u_1 & \dots & u_N \end{smallmatrix})>0$ since this is the determinant of the covariance matrix $\Cov[\bm{X},\bm{X}]$ where $\bm{X}:=(X_{u_1},\dots,X_{u_N})'$. We finally note that the point
$$
(u_1,\dots,u_N,v_1,\dots,v_N):= (u_1,\dots,u_N,u_1,\dots,u_N) \in \partial A,
$$
and by Lemma \ref{lemma_f_pos} we conclude that the integral \eqref{cb_int} is positive.
\end{proof}

\section{Discussion and concluding remarks}

\subsection{What are we hedging?}

The hedging described in Section \ref{hedging} corresponds to a `perfect' hedge of a deterministic cash flow w.r.t.\ to arbitrary shifts of the observed market rates. In many situations it is common to make simpler hedges w.r.t.\ parallel shifts of the interest rate curve. This corresponds to using modified duration as measure of interest rate risk. This is however something that should be done with caution when it comes to the Smith-Wilson method, since we know from Section \ref{parallel} that a parallel shift of a Smith-Wilson curve is inconsistent with the method itself. Moreover, from the definition of the Smith-Wilson method it is evident that the {\it entire} yield curve will depend on all market observations $\bm{p}$, or equivalently $\bm{r}$. If one still want to use a single number to describe the interest rate risk, such as modified duration, we argue that it is better to assess the modified duration w.r.t.\ to the actual market rates. That is, consider the modified duration of a Smith-Wilson discounted cash flow with respect to the {\it underlying observed spot rates}. It is straightforward to obtain expression for this quantity under Smith-Wilson discounting (calculations not included): let $\mathbf{c} = (c_1,\ldots,c_k)'$ denote the cash flow at maturities $\bm{t} = (t_1,\ldots,t_k)'$, where some of the $t_i$'s may coincide with the observed points $u_i$, then the modified duration with respect to the underlying observed spot rates, $\mathit{D}_{\mathit{sw}}(\bm{c};\bm{r})$, is given by
\begin{align}\label{mod_sw}
\mathit{D}_{\mathit{sw}}(\bm{c};\bm{r}) :&= -\frac{1}{P^c(\bm{r})}\lim_{\delta\to 0}\frac{P^c(\bm{r}+\delta) - P^c(\bm{r})}{\delta} \nonumber\\
&=\frac{\sum_{i=1}^k c_iW(t,\bm{u})\mathbf{W}^{-1}\bm{u}}{\sum_{i=1}^k c_i\left(e^{-\omega t_i} + W(t,\bm{u})\mathbf{W}^{-1}(\bm{p}-\bm{q})\right)}.
\end{align}
Note that the calculations leading up to \eqref{mod_sw} does not include a re-calibration of $\alpha$. By using \eqref{mod_sw} one get an understanding of the interest rate sensitivity of $\bm{c}$ with respect to infinitesimal parallel shifts of the underlying observed market rates, which can be used for hedging purposes. Regarding the inappropriateness of parallel shifts of a Smith-Wilson curve, we again refer the reader to Section \ref{parallel}.

\subsection{On the parametrisation of the Smith-Wilson method, choice of kernel functions and related topics}

From Section \ref{SW_process_interpretation} we know that the Smith-Wilson method can be interpreted as the expected value of a certain Gaussian process conditional on a number of {\it perfect} observations. An interesting observation is that Theorem \ref{thm_repr}, and especially Equation \eqref{P_linj}, imply that one can replace the Wilson-kernel function by {\it any} proper covariance function and still keep the process interpretation of Theorem \ref{thm_repr}. One can note that this is in fact close to what is done in the original paper \cite{SW} by Smith and Wilson where they start with the following general bond price model
\begin{align}
P(t) &:= e^{-\omega t} + \sum_{j=1}^N\zeta_jK_j(t),\label{SW_allm}
\end{align}
where $K_j(t), i = 1,\ldots,N$, denotes an arbitrary kernel function evaluated in $(t,u_j)$. Given the above process interpretation the $K_j(t)$ functions should correspond to a proper covariance function evaluated in the points $(t,u_j)$, which brings us into the realm of {\it kriging}: 

Within the area of geostatistics, there is a theory known as kriging, see e.g.\ \cite{Cr}, which in many aspects resembles the interpretation of the Smith-Wilson method as an expected value of a stochastic process. The method of kriging was introduced in {\it spatial} statistics and can be described as follows: you start by observing a number of outcomes from an unknown stochastic process (for our needs one-dimensional outcomes ordered in time), that is you assume that you make perfect observations at known locations (time points). Given these observations you want to inter/extrapolate between these known points by making assumptions on the underlying process that you have observations from. In light of this the Smith-Wilson method can be seen as one-dimensional kriging, that is Theorem \ref{thm_repr} gives us the Best Linear Unbiased Predictor (BLUP), given that we treat the theoretical unconditional expected values as {\it known a priori}.

If we return to the choice of kernel function proposed in \cite{SW}, the Wilson-function, this choice is motivated with that the Wilson-function is optimal with respect to certain regularity conditions it imposes on models of the form given by \eqref{SW_allm}. From a kriging perspective, one would instead discriminate between covariance functions by comparing the Mean Squared Error of Prediction (MSEP), where the MSEP is the conditional covariance that companions the conditional expectation given by Theorem \ref{thm_repr}. Hence, by giving the Smith-Wilson method a statistical interpretation it is possible to analyse and compare covariance functions statistically, not only w.r.t.\ MSEP.

Given the kriging representation one might be tempted to use MSEP as an alternative for calibrating $\alpha$ in the Wilson-function by e.g.\ minimising $\alpha$ at the $\cp$. This is not feasible, since the MSEP for the Wilson-function w.r.t.\ $\alpha$ lack a minimum for $\alpha>0$: by definition the MSEP is non-negative, that is
\[
\Var[Y_t] - \Cov[Y_t,\bm{Y}]\Cov[\bm{Y},\bm{Y}]^{-1}\Cov[Y_t,\bm{Y}]' > 0.
\]
Moreover, the second term is non-negative since $\Cov[\bm{Y},\bm{Y}]^{-1}$ is positive definite, and thus
\[
0 < \Var[Y_t] - \Cov[Y_t,\bm{Y}]\Cov[\bm{Y},\bm{Y}]^{-1}\Cov[Y_t,\bm{Y}]' < \Var[Y_t] \to 0,
\]
as $\alpha\to 0$, regardless of the value of $t$.

\subsection{A hedgeable Smith-Wilson method}\label{sw_hedgeable}

From Theorem \ref{thm_signs} we know that the choice of Wilson-function as kernel will result in an oscillating hedge. If we would consider changing to another kernel (or covariance)-function we want to avoid re-inventing a new un-hedgeable procedure. From the proof of Theorem \ref{thm_signs} it follows that in order for all the $\beta$'s to be positive we need conditions on the determinant of the covariance matrix associated with the covariance function. In \cite{KR} it is ascertained that all $\beta$'s will be non-negative if the inverse of the covariance matrix, i.e.\ the precision matrix, is a so-called M-matrix, see result (b) on p.\ 420 in \cite{KR}, where the definition of a M-matrix is as follows:
\begin{definition}[M-matrix, see e.g.\ \cite{KR}]
A real $n\times n$ matrix $\mathbf{A}$ s.t.\ $\mathbf{A}_{ij} \le 0$ for all $i\neq j$ with $\mathbf{A}_{ii} > 0$ is an M-matrix iff one of the following holds
\begin{enumerate}[(i)]
\item There exists a vector $\mathbf{x} \in \mathbb{R}^n$ with all positive components such that $\mathbf{Ax} > 0$.
\item $\mathbf{A}$ is non-singular and all elements of $\mathbf{A}^{-1}$ are non-negative.
\item All principal minors of $\mathbf{A}$ are positive.
\end{enumerate}
\end{definition}
From the definition of an M-matrix it follows that if the precision matrix is an M-matrix, then all components in the underlying stochastic structure will be positively associated, see result (e) on p.\ 421 in \cite{KR}. Since positive association is something observed for interest rates in practice, there might be hope to find a reasonable structure that could be used as an alternative to the Wilson-function, if wanted. Moreover, the conditions that needs to be checked are given by Theorem 8 in \cite{J}, which gives necessary and sufficient conditions under which an inverse M-matrix can be expanded and still remain an inverse M-matrices.

A simple example of a process that has an M-matrix as a precision matrix is an Ornstein-Uhlenbeck process, i.e.\ if the kernel underlying the Smith-Wilson method had been that of an Ornstein-Uhlenbeck process rather than an integrated Ornstein-Uhlenbeck, then the hedge would not be oscillating. If that choice is appropriate in other respects is beyond the scope of this paper.

\subsection{Concluding remarks}

In the present paper we have provided an alternative stochastic representation of the Smith-Wilson method (Theorem \ref{thm_repr}). This representation has {\it nothing} to do with the original derivation of the Smith-Wilson method, but it provides one stochastic representation of the method. Further, above it has been shown that the stochastic representation may be useful for interpreting and analysing the Smith-Wilson method. In particular we have shown that the method {\it always} will result in oscillating hedges w.r.t.\ to the underlying supporting market spot rates (Theorem \ref{thm_signs}). This a highly undesirable feature of the method.

Further, we have given an example where the resulting Smith-Wilson discount curve will take on negative values while fulfilling the convergence criterion, a situation which is total nonsense, but previously known to be able to occur, see e.g.\ \cite{E2,FT}. In the present paper we extend this example to show that there may also occur singularities in the convergence criterion itself and analytically show that this is a direct consequence of the occurrence of negative discount factors. That is, given that there are negative discount factors for some $\alpha \ge 0.05$ there will exist a singularity in the domain where $\alpha$ is being optimised. This is due to the fact that the $f(\cp) - \omega$ will be smaller than 0 for these values of $\alpha$, but for large enough values of $\alpha$ the discount function at $\cp$ will be positive; a situation occurring irrespective of whether or not the converged resulting Smith-Wilson curve will give rise to negative discount factors or not.

Moreover, we also provide necessary and sufficient conditions under which a change of kernel or covariance function will result in a bond price model that does not inherit the oscillating hedge behaviour, i.e.\ its inverse is an M-matrix. This is a nice feature, since given that you want an affine bond price model, you do not need to redo the entire functional analytical optimisation that initially lead up to the original Smith-Wilson method, but can merely change the kernel function and check the necessary conditions. One obvious drawback with this approach is that you from a functional analytic perspective do not know which utility function that you are optimising. The perhaps simplest process that fulfils the M-matrix criterion is the standard Ornstein-Uhlenbeck process. One can also note that the standard Ornstein-Uhlenbeck process has more than one free parameter, which intuitively may be beneficial from a modelling perspective.

We have also commented on the straightforward connection between the Smith-Wilson method and so-called kriging. This was another way to, probabilistically, be able justify different choices of alternative kernels. If one is interested in this topic, one can note that kriging in itself is a special case of so-called Bayesian non-parametrics and that kriging under certain conditions is closely connected to spline smoothing \cite{C}.

\section*{Acknowledgments}

The second author is grateful to Håkan Andersson for the discussions following the joint work \cite{AL}.

\appendix

\end{document}